\documentclass{article}
\usepackage[a4paper, portrait, margin=1.1811in]{geometry}
\usepackage[english]{babel}
\usepackage[utf8]{inputenc}
\usepackage[T1]{fontenc}
\usepackage{helvet,etoolbox,graphicx,titlesec}
\usepackage{caption,subcaption}
\usepackage{booktabs}
\usepackage[dvipsnames]{xcolor} 
\usepackage[colorlinks, citecolor=cyan]{hyperref}
\graphicspath{ {./images/} }
\usepackage[normalem]{ulem} 
\usepackage{scrextend,fancyhdr,float,here}
\usepackage{mathtools,mathrsfs,comment}
\usepackage{amsmath,amsfonts,amssymb,amsthm,epsfig,epstopdf,url,array}
\usepackage[all]{xy}
\usepackage[scr=rsfs,cal=boondox]{mathalfa}
\usepackage{tikz,pgfplots}
\pgfplotsset{compat=1.11}
\usepgfplotslibrary{fillbetween}
\usetikzlibrary{intersections}
\pgfdeclarelayer{bg}
\pgfsetlayers{bg,main}

\newtheorem{theorem}{Theorem}
\newtheorem{lemma}{Lemma}
\newtheorem{prop}{Proposition}

\newtheorem{proposition}[prop]{Proposition}

\theoremstyle{definition}\newtheorem{rem}{Remark}
\newtheorem{defi}{Definition}

\makeatletter
\patchcmd{\@maketitle}{\LARGE \@title}{\fontsize{16}{19.2}\selectfont\@title}{}{}
\makeatother
\newcommand{\vx}{\boldsymbol{x}}
\newcommand{\vk}{\boldsymbol{k}}
\newcommand{\vp}{\boldsymbol{p}}
\newcommand{\va}{\boldsymbol{\alpha}}
\newcommand{\vy}{\boldsymbol{y}}
\newcommand{\be}{\begin{equation}}
\newcommand{\ee}{\end{equation}}
\newcommand{\CCC}{\mathbb{C}}
\newcommand{\NNN}{\mathbb{N}}
\newcommand{\RRR}{\mathbb{R}}
\newcommand{\ZZZ}{\mathbb{Z}}
\newcommand{\Hilbert}{\mathscr{H}}
\newcommand{\hilbert}{\mathcal{h}}
\newcommand{\kilbert}{\mathcal{k}}
\DeclareMathOperator{\vol}{vol}
\DeclareMathOperator{\tr}{tr}
\newcommand{\sR}{\mathscr{R}}
\newcommand{\scp}[2]{\langle #1 , #2 \rangle}

\usepackage{authblk}

\newsavebox\affbox
\author[1,2]{\textbf{Pablo Costa Rico}}
\author[3]{\textbf{Roderich Tumulka}}
\affil[1]{Department of Mathematics, Technische Universit\"at M\"unchen, Germany
}
\affil[2]{Munich Center for Quantum Science and Technology (MCQST), Germany}
\affil[3]{Fachbereich Mathematik, Eberhard-Karls-Universit\"at T\"ubingen, Auf der Morgenstelle 10, 72076 T\"ubingen, Germany
}

\titlespacing\section{0pt}{12pt plus 4pt minus 2pt}{0pt plus 2pt minus 2pt}
\titlespacing\subsection{12pt}{12pt plus 4pt minus 2pt}{0pt plus 2pt minus 2pt}
\titlespacing\subsubsection{12pt}{12pt plus 4pt minus 2pt}{0pt plus 2pt minus 2pt}

\titleformat{\section}{\normalfont\fontsize{10}{15}\bfseries}{\thesection.}{1em}{}
\titleformat{\subsection}{\normalfont\fontsize{10}{15}\bfseries}{\thesubsection.}{1em}{}
\titleformat{\subsubsection}{\normalfont\fontsize{10}{15}\bfseries}{\thesubsubsection.}{1em}{}

\titleformat{\author}{\normalfont\fontsize{10}{15}\bfseries}{\thesection}{1em}{}

\title{\textbf{\huge On the Problem of Defining Charge\\[3mm] Operators for the Dirac Quantum Field}\\
	}
\date{July 12, 2024}    

\usepackage{lipsum}

\newcommand\blfootnote[1]{%
  \begingroup
  \renewcommand\thefootnote{}\footnote{#1}%
  \addtocounter{footnote}{-1}%
  \endgroup
}

\begin{document}

\pagestyle{headings}	
\newpage
\setcounter{page}{1}
\renewcommand{\thepage}{\arabic{page}}

\captionsetup[figure]{labelfont={bf},labelformat={default},labelsep=period,name={Figure }}	\captionsetup[table]{labelfont={bf},labelformat={default},labelsep=period,name={Table }}
\setlength{\parskip}{0.5em}
	
\maketitle
	
\noindent\rule{15cm}{0.5pt}
	\begin{abstract}
		It is well known how to define the operator $Q$ for the total charge (i.e., positron number minus electron number) on the standard Hilbert space of the second-quantized Dirac equation. Here we ask about operators $Q_A$ representing the charge content of a region $A\subseteq \mathbb{R}^3$ in 3d physical space. There is a natural formula for $Q_A$ but, as we explain, there are difficulties about turning it into a mathematically precise definition. First, $Q_A$ can be written as a series  but its convergence seems hopeless. Second, we show for some choices of $A$ that if $Q_A$ could be defined then its domain could not contain either the vacuum vector or any vector obtained from the vacuum by applying a polynomial in creation and annihilation operators. Both observations speak against the existence of $Q_A$ for generic $A$. 
        \\ \\
		\textbf{\textit{Keywords}}: \textit{Dirac equation; field operator; second quantization; positron; relativistic position operator}
	\end{abstract}
\noindent\rule{15cm}{0.4pt}

\blfootnote{
			 \hspace{-15pt}\small $^{*}$\textbf{Pablo Costa Rico} \textit{
				\textit{E-mail address: \color{cyan}pablo.costa@tum.de}}\\
                \small $^{*}$\textbf{Roderich Tumulka} \textit{
				\textit{E-mail address: \color{cyan}roderich.tumulka@uni-tuebingen.de}}\\ 
		}

\section{Introduction}

We are concerned here with charge operators for the Dirac quantum field in Minkowski space-time. The natural definition of the charge current density operator, given by several books on quantum field theory, is
\be
j^\mu(x) = -e  : \overline{\Psi}(x) \, \gamma^\mu \Psi(x):\,,
\ee
where $x$ is any space-time point, $-e$ the electron charge, $\Psi(x)$ the field operator, and colons mean normal ordering (Wick ordering);
see for example Schweber \cite{Schw61} Sec.~8a, Eq.~(83), Bogoliubov and Shirkov \cite{BS80} Sec.~9.1 Eq.~(9), or Geyer et al.~\cite[p.~137]{Geyer}, similarly Kaku \cite{Kaku} Eq.~(3.124). This entails in particular (for $\mu=0$) that the charge density operator at time 0 at location $\vx\in\RRR^3$ is
\be
Q(\vx) = -e \sum_{s=1}^4 : \Psi_s^*(\vx) \, \Psi_s(\vx):
\ee
(where $R^*$ means the adjoint operator of $R$ and $s$ labels the 4 dimensions of the Dirac spin space), and therefore the amount of charge contained in a region $A\subseteq \RRR^3$ is represented by the operator
\be \label{FormalChargeGeneral}
Q_A = -e \int_A \mathrm{d}\vx ~ \sum_{s=1}^4 ~ : \Psi_s^*(\vx) \, \Psi_s(\vx):
\ee
(where $\mathrm{d}\vx=\mathrm{d}x_1 \, \mathrm{d}x_2 \, \mathrm{d}x_3$ means the 3d volume element).
While for $A=\RRR^3$ it is known that one can rigorously define a self-adjoint operator $Q$ representing the total charge with spectrum $e\ZZZ$ on the standard Hilbert space of the second-quantized Dirac equation (see, e.g., \cite[Eq.~(10.52)]{Thaller} and Sec.~\ref{sec:totalcharge} below), we discuss here whether this can also be done for an arbitrary (Borel measurable) subset $A\subseteq \RRR^3$. The considerations we report here suggest that this is not possible, except in the trivial cases when either $A$ or its complement $A^c=\RRR^3\setminus A$ has volume 0 (so that either $Q_A=0$ or $Q_A=Q$). 

Studies of the Dirac quantum field \cite{DM16a,FLS} and the measurement process in relativistic quantum field theory keep attracting interest \cite{Beck,LT19,FV18,LT21, Tumulka,FP23}. Our finding is particularly relevant to the considerations of \cite{Tumulka}, where the charge operators $Q_A$ are used for the construction of a PVM on a suitable configuration space intended to represent the position operators for electrons and positrons. 

The problem does not get ameliorated if we consider ``fuzzy sets'' represented by a continuous function $g:\RRR^3\to[0,1]$ instead of the (discontinuous) characteristic function of $A$ that only assumes the values 0 and 1, as in 
\be\label{g}
Q_g = -e \int_{\RRR^3} \mathrm{d}\vx ~ g(\vx) ~ \sum_{s=1}^4 ~ :\Psi_s^*(\vx) \, \Psi_s(\vx): \,.
\ee
This point is discussed in Remark~\ref{rem:g} in Section~\ref{sec:failure}.

After reviewing the standard construction of the Dirac quantum field in Section~\ref{sec:setup}, we proceed in three steps.
\begin{enumerate}
\item In Section \ref{sec:work}, we present reasons for thinking that the rigorous interpretation of the formula \eqref{FormalChargeGeneral} should be this: Choose any orthonormal basis (ONB) $(f_j)_{j\in\NNN}$ of the subspace $\kilbert_A=L^2(A,\CCC^4)$ of functions in $L^2(\RRR^3,\CCC^4)$ that vanish outside $A$ and set
\be\label{QAdef}
Q_A := -e\sum_{j=1}^\infty : \Psi^*(f_j) \Psi(f_j) :~.
\ee

\item While, as we show in Section~\ref{sec:work}, the analog of formula \eqref{QAdef} for a finite-dimensional subspace $\kilbert$ is well-defined, self-adjoint, and independent of the choice of ONB, we explain in Section \ref{sec:failure} why the convergence of \eqref{QAdef} seems hopeless if $\vol(A)\neq 0$ and $\vol(A^c)\neq 0$.

\item In Section \ref{sec:thm1} we show for some examples of sets $A$ (in fact, for axiparallel cubes of side length $2\pi$) that even if formula \eqref{QAdef} could be made sense of, and $Q_A$ could be rigorously defined, it would have the unpleasant property that its domain $D_A$ contains neither the vacuum vector $|\Omega\rangle$, nor any vector obtained from $|\Omega\rangle$ by applying finitely many creation and annihilation operators, nor any finite linear combination of such vectors. More precisely, let 
\be\label{F0def}
F_0:= \bigcup_{m=0}^\infty \bigcup_{n=0}^\infty\Bigl\{ \Psi^*(g_1)\cdots \Psi^*(g_m)\Psi(h_1)\cdots\Psi(h_n)|\Omega\rangle \Big| \text{all }g_i,h_j\in L^2(\RRR^3,\CCC^4)\Bigr\}
\ee
let $F$ be the span (i.e., the set of finite complex-linear combinations) of $F_0$ in the Hilbert space of the Dirac quantum field (see Section \ref{sec:setup} for the definition), and let
\be\label{Qtruncated}
Q_{A,(f_j)}^J := \sum_{j=1}^J :\Psi^*(f_j) \Psi(f_j) :
\ee
denote the \emph{truncated} charge operator based on the ONB $(f_j)_j$ of $L^2(A,\CCC^4)$. 
\end{enumerate}

\begin{theorem}\label{TheoremVacuum}
Let $0\neq\psi\in F$, let $\vx\in\RRR^3$, and let $A=\vx+[-\pi,\pi]^3\subset \RRR^3$. Then there is an ONB $(f_j)_{j\in\NNN}$ of $L^2(A,\CCC^4)$ such that
\be
\lim_{J\to\infty} \Bigl\| Q_{A,(f_j)}^J ~ \psi\Bigr\|^2  = \infty\,.
\ee
\end{theorem}

We give the proof in Section \ref{sec:thm1}.

\begin{rem}
One might think of defining the charge operator $Q_A$ in the following different way instead of \eqref{QAdef} (although, as far as we know, nobody actually proposed this definition):
\begin{equation}
    \tilde{Q}_A:=-e\sum_{j=1}^{\infty}\Bigl( b^*( f_j)b(f_j)-c^*( f_j)c( f_j) \Bigr),
\end{equation}
where  $(f_j)_{j \in \mathbb{N}}$ is again an ONB of $L^2(A,\mathbb{C}^4)$, and $b$ and $c$ are the electron and positron annihilation operators as in $\Psi(f)=b(f)+c^*(f)$ (see Section~\ref{sec:setup}). In fact, this operator is independent of the choice of ONB of $L^2(A,\CCC^4)$, it is well defined, self-adjoint, and additive in $A$,
\be
\tilde{Q}_{A_1\cup A_2}= \tilde{Q}_{A_1} + \tilde{Q}_{A_2}
\ee
whenever $A_1\cap A_2=\emptyset$; all of these are desired properties of charge operators. However, the spectrum $\sigma(\tilde{Q}_A)$ does not necessarily contain only integer multiples of the electron charge $-e$, and $\tilde{Q}_{A_1}$ and $\tilde{Q}_{A_2}$ do not in general commute with each other, even if $A_1$ and $A_2$ are disjoint; these properties speak against regarding $\tilde{Q}_A$ as the charge operator for the region $A$ and inhibit their use in the way proposed in \cite{Tumulka}.
\end{rem}

\section{Setup}
\label{sec:setup}

In this section, we briefly review the definition of the Hilbert space and field operators for the Dirac quantum field following \cite{Thaller}.

\subsection{Hilbert Space}

In the 1-particle Hilbert space $\mathcal{h}:=L^2(\RRR^3,\CCC^4)$, the 1-particle Hamiltonian is the free Dirac operator
\be
h_D = -i \sum_{a=1}^3 \gamma^0 \gamma^a \partial_a + m\gamma^0
\ee 
with $m\geq 0$ the mass and $c=1=\hbar$. It is self-adjoint with spectrum
\be
\sigma(h_D) = (-\infty,-m] \cup [m,\infty)\,,
\ee
so $\mathcal{h}$ can be decomposed into the orthogonal sum of the positive spectral subspace $\mathcal{h}_+$ and the negative spectral subspace $\mathcal{h}_-$,
\begin{equation}\label{DecompositionPosNeg}
    \mathcal{h}=\mathcal{h}_+ \oplus \mathcal{h}_-~.
\end{equation}
We write $P_\pm$ for the projection $\mathcal{h}\to \mathcal{h}_\pm$.
The \emph{charge conjugation operator}
$C$ \cite[p.~14]{Thaller} is an anti-unitary, anti-linear operator $\mathcal{h}\to \mathcal{h}$, given in the standard representation by
\begin{equation}
    Cf=i  \gamma^2\overline{f},
\end{equation}
where the bar means complex conjugation. $C$ is an involution, $C^2=I$, and maps $\mathcal{h}_+$ to $\mathcal{h}_-$ and vice versa, $C\mathcal{h}_\pm = \mathcal{h}_\mp$. 

The Hilbert space of the Dirac quantum field is \cite[Sec.~10.1]{Thaller} 
\begin{align}
\Hilbert
&=\mathcal{F}_-(\mathcal{h}_+)\otimes\mathcal{F}_-(C\mathcal{h}_-)\\
&=\bigoplus_{n,m=0}^\infty \mathcal{h}_+^{\wedge n}\otimes (C\mathcal{h}_-)^{\wedge m}\,,
\end{align}
where $\mathcal{F}_-$ is the fermionic Fock space \cite[Chap. 2]{AlgebraicQFT}, \cite[Sec. 10.1]{Thaller}, $\wedge n$ means the $n$-th anti-symmetrized tensor power, and $C\mathcal{h}_-$ is actually equal to $\mathcal{h}_+$. The \emph{vacuum vector} $\Omega\in\Hilbert$ is the one with $\Omega^{(0,0)}=1$ and $\Omega^{(n,m)}=0$ for $(n,m)\neq (0,0)$.


\subsection{Field Operators}

For $f \in \mathcal{h}$ and $\psi=\displaystyle \bigoplus_{n,m=0}^{\infty} \psi^{(n,m)} \in \Hilbert$, the creation operator for particles is defined as
\begin{equation}
(b^*(f)\psi)^{(n,m)}=\frac{1}{\sqrt{n}}\sum_{j=1}^{n}(-1)^{j+1}(P_+ f)\otimes_j^\mathrm{p}\psi^{(n-1,m)},
\end{equation}
where $g\otimes_j (f_1\otimes \cdots \otimes f_{n-1}):= f_1\otimes \cdots \otimes f_{j-1}\otimes g \otimes f_{j} \otimes \cdots \otimes f_{n-1}$, and the tensor product $\otimes_j^\mathrm{p}$ (p for ``particles'') only affects the first $n$ arguments. The creation operator for anti-particles is defined as 
\begin{equation}
        (c^*(f)\psi)^{(n,m)}=\frac{(-1)^n}{\sqrt{m}}\sum_{j=1}^{m}(-1)^{j+1}(CP_- f)\otimes_j^\mathrm{ap}\psi^{(n,m-1)},
\end{equation}
where the tensor product $\otimes_j^\mathrm{ap}$ (ap for ``anti-particles'') only affects the last $m$ arguments. 
The annihilation operator for particles is 
\begin{equation}
    (b(f)\psi)^{(n,m)}=\sqrt{n+1}\langle P_+ f, \psi^{(n+1,m)} \rangle_{1\mathrm{p}},
\end{equation}
where $\langle \hspace{3pt}, \hspace{3pt} \rangle_{1\mathrm{p}}: \mathcal{h}\times \mathcal{h}^{\otimes (n+1)} \to   \mathcal{h}^{\otimes n}$ is the partial inner product in the first argument, i.e., the sesquilinear extension of 
\begin{equation}\label{innerannihilation}
(\psi, \varphi_1 \otimes \cdots \otimes \varphi_{n+1})  \mapsto  \langle \psi, \varphi_1 \rangle \varphi_2\otimes \cdots \otimes \varphi_{n+1}\,,
\end{equation}
here restricted to the space $\mathcal{h}_+^{\wedge n}$. In the same way, the annihilation operator for antiparticles is
\begin{equation}
    (c(f)\psi)^{(n,m)}=(-1)^n\sqrt{m+1}\langle CP_- f, \psi^{(n,m+1)} \rangle_{1\mathrm{ap}},
\end{equation}
where  $\langle \hspace{3pt}, \hspace{3pt} \rangle_{1\mathrm{ap}}$ is the linear extension of \eqref{innerannihilation} restricted to $(C\mathcal{h}_-)^{\wedge m}$.
With this, the field operator $\Psi(f)$ is defined \cite{Thaller} as
\begin{equation}\label{Psidefbc}
    \Psi(f)=b(f)+c^*(f),
\end{equation}
and its adjoint as
\begin{equation}\label{Psi*defbc}
    \Psi^*(f)=b^*(f)+c(f),
\end{equation}
which satisfy the anticommutation relations
\begin{equation}
    \{\Psi(f),\Psi(g)\}=0=\{\Psi^*(f),\Psi^*(g)\}, \quad \{\Psi^*(f),\Psi(g)\}=\langle g,f \rangle id,
\end{equation}
for every $f,g \in \mathcal{h}$, while 
\begin{subequations}\label{bcanticommute}
\begin{align}
\{b(f),b(g)\}&=0=\{b^*(f),b^*(g)\}\,,~\{b^*(f),b(g)\} = \langle g,P_+f\rangle id,\\
\{c(f),c(g)\}&=0=\{c^*(f),c^*(g)\}\,,~\{c^*(f),c(g)\} = \langle g,P_-f\rangle^* id,\label{c*ccommute}\\
\{b^{\#}(f),c^{\%}(g)\}&=0\,,
\end{align}
\end{subequations}
where $\#$ and $\%$ mean either $*$ or no $*$.

The \emph{normal order} or \emph{Wick order} of an expression
\be
\sum_j \alpha_j A_{j1}\cdots A_{jr_j}
\ee
with finitely or infinitely many terms, complex coefficients $\alpha_j$, and each symbol
\be
A_{ji}\in \bigl\{b(f),b^*(f),c(f),c^*(f):f\in\hilbert\bigr\}
\ee
is defined to be
\be
:\sum_j \alpha_j A_{j1}\cdots A_{jr_j}:~~ 
=\sum_j \alpha_j ~:A_{j1}\cdots A_{jr_j}:~~
=\sum_j (-1)^{\rho_j} \alpha_j A_{j\rho_j(1)}\cdots A_{j\rho_j(r_j)}\,,
\ee
where $(-1)^{\rho}$ means the sign of the permutation $\rho$, and $\rho_j$ means any permutation such that in the expression $A_{j\rho_j(1)}\cdots A_{j\rho_j(r_j)}$, there is no $b^*(f)$ to the right of any $b(g)$ and no $c^*(f)$ to the right of any $c(g)$. (It is well known that for any two such permutations $\rho_j$, the expressions denote the same operator by virtue of the anticommutation relations \eqref{bcanticommute}.)

For expressions in the field operators (i.e., with $A_{ji}\in \{\Psi(f),\Psi^*(f):f\in\hilbert\}$), normal order is carried out after using \eqref{Psidefbc} and \eqref{Psi*defbc} to express $\Psi(f)$ and $\Psi^*(f)$ in terms of $b(f),b^*(f),c(f)$, and $c^*(f)$. Thus,
\begin{subequations}\begin{align}
:\Psi^*(f) \, \Psi(f): ~~
&= ~~:\bigl(b^*(f)+c(f)\bigr)\bigl(b(f)+c^*(f)\bigr):\\
&= ~~:\Bigl(b^*(f)b(f)+c(f)b(f)+b^*(f)c^*(f)+c(f)c^*(f)\Bigr):\\
&= ~~~~\;b^*(f)b(f)+c(f)b(f)+b^*(f)c^*(f)-c^*(f)c(f)\\
&= ~~\Psi^*(f) \, \Psi(f) - \|P_-f\|^2\label{Wick}
\end{align}
\end{subequations}
by \eqref{c*ccommute}.

\section{Cases That Work}
\label{sec:work}

Henceforth, we drop the factor $-e$ from the charge operators \eqref{QAdef}. In this section, we study expressions of the form
\be\label{QKdef}
Q_\kilbert := \sum_{j} : \Psi^*(f_j) \Psi(f_j) :~,
\ee
where $(f_j)_j$ is an ONB of the subspace $\kilbert\subseteq L^2(\RRR^3,\CCC^4)$. We show that they are well defined and well behaved if $\dim\kilbert<\infty$; in Section~\ref{sec:failure} we will address the convergence problems in the case $\dim\kilbert=\infty$.

\subsection{The $Q_\kilbert$ Operators}

Although we are ultimately interested in the $\infty$-dimensional subspace $\kilbert_A=L^2(A,\CCC^4)$ with $A\subset \RRR^3$ and $\vol(A)\neq 0$, $\vol(A^c)\neq 0$, we first focus on finite-dimensional ones. We verify that they have several properties that would appear reasonable for the desired $Q_A$ operators.

\begin{prop}\label{prop:finite}
Let $\kilbert\subset \mathcal{h}$ be a subspace of finite dimension $d$ and $f_1,\ldots,f_d$ an ONB of $\kilbert$. Then \eqref{QKdef} defines an operator $Q_\kilbert$ that is independent of the choice of ONB. It is bounded and self-adjoint with spectrum
\be\label{SpectrumQk}
\sigma(Q_\kilbert) = \{-d^-,-d^- +1, -d^- +2,\ldots, -d^- + d\}
\ee
with
\be
d^\pm = \tr(P_\kilbert P_\pm)\,.
\ee
\end{prop}

We give all proofs in Section~\ref{sec:proofs}. In analogy to the situation of Proposition 1, we would hope that $Q_A$ is self-adjoint with spectrum $\ZZZ$ because the amount of charge in $A$ should be an integer. 

\begin{prop}\label{prop:finiteperpQcommute}
If $\kilbert_1 \perp\kilbert_2$ are finite-dimensional subspaces of $\mathcal{h}$, then
\be
[Q_{\kilbert_1},Q_{\kilbert_2}]=0
\ee
and
\be\label{Qkilbertadd}
Q_{\kilbert_1\oplus \kilbert_2} = Q_{\kilbert_1}+Q_{\kilbert_2}\,.
\ee
\end{prop}

Note that if $A_1\cap A_2=\emptyset$, then $L^2(A_1,\CCC^4)\perp L^2(A_2,\CCC^4)$ and $L^2(A_1\cup A_2,\CCC^4)=L^2(A_1,\CCC^4) \oplus L^2(A_2,\CCC^4)$. We would want in that case that $[Q_{A_1},Q_{A_2}]=0$ and $Q_{A_1\cup A_2}=Q_{A_1} + Q_{A_2}$ (charge content is additive).

\begin{prop}\label{prop:finiteQcommute}
If $\kilbert_1,\kilbert_2$ are finite-dimensional subspaces of $\mathcal{h}$ and $\kilbert_1=\mathcal{A}\oplus \mathcal{B}$, $\kilbert_2=\mathcal{A}\oplus \mathcal{C}$ with $\mathcal{A}, \mathcal{B}, \mathcal{C}$ mutually orthogonal, then 
\be
[Q_{\kilbert_1},Q_{\kilbert_2}]=0\,.
\ee
\end{prop}

Note that if $A_1,A_2 \subset \RRR^3$ and $\kilbert_i=L^2(A_i,\CCC^4)$ (which of course is $\infty$-dimensional), then $\kilbert_1=\mathcal{A}\oplus \mathcal{B}$, $\kilbert_2=\mathcal{A}\oplus \mathcal{C}$ with $\mathcal{A}, \mathcal{B}, \mathcal{C}$ mutually orthogonal, as $\mathcal{A}=L^2(A_1\cap A_2,\CCC^4)$, $\mathcal{B}=L^2(A_1\setminus A_2,\CCC^4)$, $\mathcal{C}=L^2(A_2\setminus A_1,\CCC^4)$. Correspondingly, we would also expect of the desired $Q_A$ operators that $[Q_{A_1},Q_{A_2}]=0$.

We also observe that $L^2(A,\CCC^4)$ is invariant under the charge conjugation operator $C$,
\be\label{ACinv}
CL^2(A,\CCC^4)=L^2(A,\CCC^4).
\ee

\begin{prop}\label{prop:d2}
If $\kilbert\subset \mathcal{h}$ is finite-dimensional and $C\kilbert=\kilbert$, then $d^+=d^-=d/2$, so $\sigma(Q_{\kilbert})=\{-\tfrac{d}{2},-\tfrac{d}{2}+1,\ldots, \tfrac{d}{2}-1, \tfrac{d}{2}\}$. In particular, $\sigma(Q_{\kilbert})\subset \ZZZ$ for even $d$ and $\sigma(Q_{\kilbert})\subset \tfrac{1}{2}+\ZZZ$ for odd $d$.
\end{prop}

These facts about $Q_\kilbert$ motivate us to consider \eqref{QAdef} as the definition of $Q_A$ and to study whether the definition \eqref{QKdef} of $Q_\kilbert$ can be extended to $\infty$-dimensional $\kilbert$.

\subsection{The Total Charge Operator}
\label{sec:totalcharge}

For $A=\RRR^3$, the desired operator would represent the total charge. Such an operator is well known to exist \cite[Sec.~10.2.3]{Thaller}, it can be written as
\be\label{Qdef}
Q= \sum_j \Bigl( b^*(f_j)b(f_j)- c^*(f_j) c(f_j) \Bigr)\,,
\ee
where $(f_j)_j$ is the union of an ONB of $\hilbert_+$ and one of $\hilbert_-$ in some order, it is independent of the choice of these ONBs and the overall ordering, it is self-adjoint with spectrum
\be
\sigma(Q)=\ZZZ\,,
\ee
and its eigenspace with eigenvalue $q\in\ZZZ$ is the sum of all $(m,n)$-sectors of $\Hilbert$ with $n-m=q$. It is easy to see that $Q$ can also be written in the form $\sum_j :\Psi^*(f_j) \Psi(f_j):$ if, as before, each $f_j$ lies in either $\hilbert_+$ or $\hilbert_-$. This is another reason  that would seem to support the tentative definition \eqref{QAdef} of $Q_A$.

We may also note in passing that if $\vol(A)=0$, then $\kilbert=L^2(A,\CCC^4)=\{0\}$, and \eqref{QKdef} yields the reasonable outcome $Q_A=0$. Likewise, if $\vol(A^c)=0$, then $\kilbert=\mathcal{h}$, and \eqref{QKdef} yields, if each $f_j$ lies in either $\hilbert_+$ or $\hilbert_-$, the reasonable outcome $Q_A=Q_{\RRR^3}=Q$.

\subsection{Subspaces Aligned with $\hilbert_\pm$}

We briefly turn to another family of subspaces $\kilbert\subseteq \hilbert$ for which we can show that \eqref{QKdef} actually defines an operator (at least for certain bases $(f_j)_j$): these are the $\kilbert$ aligned with $\hilbert_\pm$, i.e., for which the angle between $\kilbert$ and $\hilbert_\pm$ is $90^\circ$ (like between the $xy$ and $xz$ planes in $\RRR^3$) or, put differently, $[P_\kilbert,P_\pm]=0$ or, put again differently, $\kilbert$ is an orthogonal sum $\kilbert=\kilbert_+ \oplus \kilbert_-$ with $\kilbert_\pm \subseteq \hilbert_\pm$. Here, $\dim \kilbert$ may be infinite, and the total charge operator $Q$ of Section~\ref{sec:totalcharge} is included as the special case $\kilbert=\hilbert$. (Of course, $\kilbert_A=L^2(A,\CCC^4)$ is in general not aligned with $\hilbert_\pm$ because $\kilbert_A\cap \hilbert_\pm =\{0\}$ whenever both $A$ and $A^c$ have non-empty interior, as no nonzero function from $\hilbert_\pm$ can vanish on a non-empty open set \cite[Cor.~1.7]{Thaller}.)

So let $\kilbert$ be aligned with $\hilbert_\pm$ and  let $(f_j)_j$ be an ONB of $\kilbert$ so that every $f_j$ belongs to either $\kilbert_+$ or $\kilbert_-$. If $f_j\in \kilbert_+$, then
\be
:\Psi^*(f_j)\, \Psi(f_j): ~~=~~ b^*(f_j) \, b(f_j).
\ee
If $f_j\in \kilbert_-$, then 
\be
:\Psi^*(f_j)\, \Psi(f_j): ~~=~ -c^*(f_j) \, c(f_j).
\ee
It follows that the series \eqref{QKdef} converges, in fact in any order, to 
\be
Q_\kilbert = N_{\kilbert_+}-N_{\kilbert_-}\,,
\ee
where $N_{\kilbert_\pm}$ is the number operator associated with the appropriate subspace. Moreover, $Q_\kilbert$ is a self-adjoint operator with spectrum in $\ZZZ$. (This confirms also the statements made in Section~\ref{sec:totalcharge} about the total charge operator.)

\section{Failure of Convergence}
\label{sec:failure}

We now look at the question of convergence of the series \eqref{QAdef} that we consider as the definition of $Q_A$ and explain why there is not much hope of finding a sense in which it might converge.

Let $(f_j)_j$ be again an ONB of $L^2(A,\CCC^4)$ with $\vol(A)\neq 0 \neq \vol(A^c)$. As we will see, the $j$-th term in the series \eqref{QAdef},
\be
T_j:= ~~:\Psi^*(f_j) \, \Psi(f_j):~~,
\ee
is a bounded self-adjoint operator on $\Hilbert$, and all $T_j$ commute with each other,
\be
[T_i,T_j]=0 ~~\forall i,j\in\NNN.
\ee
It follows \cite[Chap.~2 Thm.~1.3, p.~122]{Ber86} that they can be diagonalized simultaneously through a projection-valued measure (PVM) $P$ on a subset $\sR$ of $\RRR^\NNN$ (equipped with the $\sigma$-algebra generated by the cylinder sets). Each $T_j$ can then be expressed in terms of $P$ like a multiplication operator,
\be
T_j = \int_{\sR}P(\mathrm{d}\lambda)\, \tau_j(\lambda) 
\ee
with $\tau_j:\sR\to\sigma(T_j)$. This suggests that the series $\sum_{j=1}^\infty T_j$ converges in some sense if and only if the series
\be\label{tauseries}
\sum_{j=1}^\infty \tau_j
\ee
converges in a suitable sense, and we now explain why it seems implausible that \eqref{tauseries} could converge in any sense. 

As a preparation, we note that $T_j$ has only two eigenvalues, which are $\|P_+f_j\|^2$ and $-\|P_-f_j\|^2$, and can be written as
\be
T_j = \|P_+f_j\|^2 P_j -\|P_-f_j\|^2 P_j^\perp\,,
\ee
where
\be
P_j = \Psi^*(f_j) \, \Psi(f_j)
\ee
is a projection (as we verify in the proof of Proposition~\ref{prop:finite} in Section~\ref{sec:proofs}), and $P_j^\perp=I-P_j$ is the projection to the orthogonal complement of the range of $P_j$. Put differently, $P_j$ and $P_j^\perp$ are the eigenprojections to the two eigenvalues of $T_j$.

It was part of the hope that $Q_A$ should not depend on the choice of the ONB $(f_j)_j$. We can simplify the situation further by choosing $(f_j)_j$ in a special way:

\begin{prop}\label{prop:C}
Let $\kilbert$ be any separable Hilbert space (of finite or infinite dimension) and $C:\kilbert\to\kilbert$ any anti-unitary operator that is an involution, $C^2=I$. Then there is an ONB $(f_j)_j$ of $\kilbert$ such that the action of $C$ is just complex conjugation of the coefficients,
\be\label{Cfj}
C\Bigl( \sum_j c_j \, f_j \Bigr) = \sum_j c_j^* \, f_j
\ee
for arbitrary $c_j\in\CCC$ with $\sum_j |c_j|^2 < \infty$.
\end{prop}

In particular, if $\kilbert$ is a closed subspace of $\hilbert=L^2(\RRR^3,\CCC^4)$ that is invariant under the charge conjugation operator $C$ (such as $\kilbert=L^2(A,\CCC^4)$), then there is an ONB $(f_j)_j$ of $\kilbert$ with $Cf_j=f_j$ for each $j$.

For such a basis, it follows that
\be
\|P_+f_j\|^2 = \frac{1}{2} = \|P_-f_j\|^2.
\ee
Indeed, using that $CP_-C=P_+$ and $CC=I$, we find that 
\be
CP_-f_j = CP_-CCf_j = P_+ Cf_j= P_+f_j\,,
\ee
and since $C$ is anti-unitary,
\be
\|P_-f_j\|^2=\|CP_-f_j\|^2=\| P_+ f_j\|^2
\ee
while
\be
\| P_+ f_j\|^2 +\|P_- f_j\|^2=\|f_j\|^2=1\,.
\ee

For the convergence question, this means that
\be
T_j = \tfrac12 P_j -\tfrac12 P_j^\perp\,.
\ee
Now we return to the convergence of \eqref{tauseries}. At every point in $\sR$, this series has all summands equal to either $\tfrac12$ or $-\tfrac12$, and therefore cannot converge. (If it did, then for every $\varepsilon>0$ (such as $\varepsilon=1/5$), the sequence of partial sums would have to stay, from some term onward, within the $\varepsilon$-neighborhood of the limit, but that is impossible if the partial sum changes with the next term by $1/2$.) 

\begin{rem}
    In particular, for the case $A=\mathbb{R}^3$ this argument suggests that the operator $Q_A$ depends on the ONB, since the series \eqref{QAdef} converges (to $Q$ as in \eqref{Qdef}) if $(f_j)_j$ is the union of an ONB of $\hilbert_+$ and one of $\hilbert_-$, but the argument suggests that the series does not converge if $(f_j)_j$ is chosen so that $Cf_j=f_j$.
\end{rem}

\begin{rem}
    What if the convergence of the series \eqref{QAdef}  and/or its limit \emph{did} depend on the choice of basis $(f_j)_j$? Then the basis most relevant to us would be, not an \emph{actual} ONB but a \emph{generalized} ONB, the \emph{position basis} $\{|\vx,s\rangle: \vx\in\RRR^3,s=1,2,3,4\}$. And the analog of $\|P_\pm f_j\|^2$ would be $\langle \vx,s|P_\pm|\vx,s \rangle$.
\end{rem}

\begin{rem}\label{rem:g}
    How would these considerations change if, instead of a subset $A\subseteq \RRR^3$, we considered a function $g:\RRR^3\to[0,1]$ (a ``fuzzy set'') as in \eqref{g}? Then multiplication by $g$ is a bounded self-adjoint operator $M_g:\hilbert\to\hilbert$ that commutes with $C$. Suppose first that $g$ assumes only finitely many values, so $M_g$ has finite spectrum, and there is an ONB $(f_j)_j$ of $\hilbert$ diagonalizing it, 
    \be
    M_g=\sum_j m_j |f_j\rangle \langle f_j|,
    \ee
    with 
    \be
    \sum_{j=1}^\infty m_j = \tr(M_g) = \infty\,,
    \ee
    unless $g$ vanishes almost everywhere.
    Then (dropping again the charge factor $-e$), the expression for the operator $Q_g$ is
    \be
        Q_g = \sum_{j=1}^\infty m_j \, : \Psi^*(f_j)\Psi(f_j):~,
    \ee
    and the convergence of this series would face the same difficulties as that of \eqref{QAdef} because $\sum_j m_j =\infty$. 

    Now suppose that $g$ assumes infinitely many values, including the case that $g$ is continuous and assumes the values 0 and 1. Then $g$ could be approximated by a sequence of functions $g_n:\RRR^3\to[0,1]$ such that $g_n$ assumes only finitely many values, always $g_n\leq g_{n+1}\leq g$, and $g_n\to g$ uniformly as $n\to \infty$ (e.g., let $g_n(\vx)$ be the largest integer multiple of $2^{-n}$ not greater than $g(\vx)$). Since each $g_n$ is nowhere further from 0 than $g$, it seems that the chances of $Q_g$ being well defined are no greater than those for $Q_{g_n}$, but already for the latter the series representation seems to fail to converge.
\end{rem}

\section{Proofs}
\label{sec:proofs}

\subsection{Proofs of Propositions}
\label{sec:pfprop}

\begin{proof}[Proof of Proposition~\ref{prop:finite}]
Since $\Psi(f)$ and $\Psi^*(f)$ are, for every $f\in\hilbert$, well defined and (due to the fermionic symmetry) bounded, so is $\Psi^*(f)\Psi(f)$ and, by \eqref{Wick}, $:\Psi^*(f)\Psi(f):$ . Since in the present case the sum in \eqref{QKdef} has finitely many terms, $Q_\kilbert$ is well defined and bounded. Since $\Psi^*(f)$ is the adjoint of $\Psi(f)$, $\Psi^*(f)\Psi(f)$ is self-adjoint and, by \eqref{Wick}, so is $Q_\kilbert$.

We turn to the basis-independence. Let $(g_j)_{j=1\ldots d}$ be another ONB of $\kilbert$ and $U_{ij}=\scp{f_i}{g_j}$ the unitary $d\times d$ matrix mapping the $f$'s to the $g$'s, $g_j=\sum_i U_{ij} f_i$. Then, by anti-linearity of $\Psi(f)$ in $f$, linearity of $\Psi^*(f)$ in $f$, and unitarity of $U$,
\begin{subequations}
\begin{align}
    \sum_{j=1}^d :\Psi^*(g_j)\Psi(g_j):~
    &\stackrel{\eqref{Wick}}{=}\sum_{j=1}^d \Bigl(\Psi^*(g_j) \Psi(g_j) - \|P_- g_j\|^2\Bigr)\\
    &=\sum_{j=1}^d \Psi^*\Bigl(\sum_{i=1}^d U_{ij}f_i\Bigr)\: \Psi\Bigl(\sum_{\ell=1}^d U_{\ell j}f_\ell\Bigr) -\sum_{j=1}^d \scp{g_j}{P_-g_j}\\
    &=\sum_{ij\ell=1}^d  U_{ij} U^*_{\ell j} \Psi^*(f_i) \Psi(f_\ell) - \tr(P_\kilbert P_-)\\
    &=\sum_{i\ell=1}^d \delta_{i\ell} \Psi^*(f_i) \Psi(f_\ell) - \sum_{i=1}^d \scp{f_i}{P_-f_i}\\
    &=\sum_{i=1}^d \Bigl(\Psi^*(f_i)\Psi(f_i) -\|P_-f_i\|^2\Bigr)\\
    &\stackrel{\eqref{Wick}}{=}\sum_{i=1}^d :\Psi^*(f_i)\Psi(f_i):~.
\end{align}
\end{subequations}

Concerning the spectrum, since $\sum_j \|P_-f_j\|^2=\tr(P_\kilbert P_-)=d^-$, it suffices to show that 
    \begin{equation}
        \sigma\Biggl(  \sum_{j=1}^d \Psi^*(f_j)\Psi(f_j)\Biggr)=\{0,1, \hdots,d\}.
    \end{equation}
Note that $P_j:= \Psi^*(f_j)\Psi(f_j)$ is a projection operator because it is self-adjoint and, by the canonical anti-commutation relations (CAR),
\begin{subequations}
\begin{align}
P_j^2 &= \Psi^*(f_j)\Psi(f_j)\Psi^*(f_j)\Psi(f_j)\\
&=\Psi^*(f_j)\Bigl(-\Psi^*(f_j)\Psi(f_j)+\underbrace{\|f_j\|^2}_{=1} \Bigr)\Psi(f_j)\\
&= 0+\Psi^*(f_j)\Psi(f_j)=P_j
\end{align}
\end{subequations}
because, by the CAR again, $\Psi(f)\Psi(f)=0=\Psi^*(f)\Psi^*(f)$. Moreover, the $P_j$ commute with each other because $\Psi(f_i)$ and $\Psi^*(f_j)$ anti-commute (as $\scp{f_i}{f_j}=0$), so
\begin{subequations}
\begin{align}
    P_iP_j 
    &= \Psi^*(f_i)\Psi(f_i)\Psi^*(f_j)\Psi(f_j)\\
    &= \Psi^*(f_i)\Bigl( -\Psi^*(f_j)\Psi(f_i)\Bigr) \Psi(f_j)\\ &=\Bigl(\Psi^*(f_j)\Psi^*(f_i)\Bigr)\Psi(f_i)\Psi(f_j)\\
    &=\Psi^*(f_j)\Psi^*(f_i)\Bigl(-\Psi(f_j)\Psi(f_i)\Bigr)\\
    &=\Psi^*(f_j)\Bigl(\Psi(f_j)\Psi^*(f_i)\Bigr) \Psi(f_i) = P_jP_i\,.
\end{align}
\end{subequations}

Since a projection has eigenvalues 0 and 1, a sum of $d$ commuting projections can have no other eigenvalues than $0,1,\ldots,d$. In order to see that all of these values occur, we use that $P_j\Psi(f_j)=0$, and then
\begin{equation}
    (P_1+\ldots+P_d)\Psi(f_d)\cdots \Psi(f_1)\Omega=0.
\end{equation}
In general, for the eigenvalue $j$, consider for example the eigenvector
\begin{equation}    \varphi_j=\Psi(f_d)\cdots\Psi(f_{j+1})\Psi^*(f_j)\cdots \Psi^*(f_1)\Omega,
\end{equation}
and use that $P_j\Psi^*(f_j)=\Psi^*(f_j)$.
\end{proof}

\begin{proof}[Proof of Proposition~\ref{prop:finiteperpQcommute}]
    Let $(f_1\ldots,f_{d_1})$ be an ONB of $\kilbert_1$ and $(g_1,\ldots,g_{d_2})$ one of $\kilbert_2$. Then
    \begin{align}
        [Q_{\kilbert_1},Q_{\kilbert_2}]
        &= \sum_{j=1}^{d_1}\sum_{\ell=1}^{d_2} \Bigl[ \Psi^*(f_j)\Psi(f_j)-\|P_-f_j\|^2, \Psi^*(g_\ell)\Psi(g_\ell)-\|P_-g_\ell\|^2\Bigr]\\
        &=\sum_{j=1}^{d_1}\sum_{\ell=1}^{d_2} \Bigl[ \Psi^*(f_j)\Psi(f_j), \Psi^*(g_\ell)\Psi(g_\ell)\Bigr]\\
        &=\sum_{j=1}^{d_1}\sum_{\ell=1}^{d_2} \Bigl( \Psi^*(f_j)\Psi(g_\ell) \scp{f_j}{g_\ell} - \Psi^*(g_\ell)\Psi(f_j) \scp{g_\ell}{f_j} \Bigr)=0
    \end{align}
    since $\kilbert_1\perp\kilbert_2$. Finally, Eq.~\eqref{Qkilbertadd} follows from the fact that $(f_1,\ldots,f_{d_1},g_1,\ldots,g_{d_2})$ is an ONB of $\kilbert_1\oplus \kilbert_2$.
\end{proof}

\begin{proof}[Proof of Proposition~\ref{prop:finiteQcommute}]
    By Proposition~\ref{prop:finiteperpQcommute}, 
    \be
    Q_{\kilbert_1}=Q_\mathscr{A}+Q_\mathscr{B} \text{ and }
    Q_{\kilbert_2}=Q_\mathscr{A}+Q_\mathscr{C}\,. 
    \ee
    Therefore,
    \be
    [Q_{\kilbert_1},Q_{\kilbert_2}] =[Q_\mathscr{A},Q_\mathscr{A}]+[Q_\mathscr{B},Q_\mathscr{A}]+[Q_\mathscr{A},Q_\mathscr{C}]+[Q_\mathscr{B},Q_\mathscr{C}]=0
    \ee
    by Proposition~\ref{prop:finiteperpQcommute} again, since $\mathscr{A},\mathscr{B},\mathscr{C}$ are mutually orthogonal.
\end{proof}

\begin{proof}[Proof of Proposition~\ref{prop:d2}]
    The condition $C\kilbert=\kilbert$ translates into $P_{\kilbert}CP_{\kilbert}=CP_{\kilbert}$. Taking the anti-linear adjoint we obtain that $P_{\kilbert}CP_{\kilbert}$ =  $P_{\kilbert}C$, which implies that $CP_{\kilbert}=P_{\kilbert}C$. Since $CP_-C=P_+$, if we trace in $\kilbert$ we finally obtain
    \begin{equation}
        d^+=\tr_{\kilbert}(P_+)=\tr(CP_-CP_{\kilbert})=\tr(P_-C^2P_{\kilbert})=\tr_{\kilbert}(P_-)=d-d^+\,.
    \end{equation}
    The case for $d^-$ is analogous.
\end{proof}

We now turn to Proposition~\ref{prop:C}. The statement was already given by Wigner in \cite[Eq.~(18)]{Wig60} and Ruotsalainen in \cite[Prop.~2.3]{Ruo12}, but the proofs there have a gap in the infinite-dimensional case,\footnote{They both construct an orthonormal sequence $v_{11},v_{12},\ldots$ of vectors invariant under $A=C$, but it is not clear that (as they imply) the closed span of this sequence equals $\kilbert$ (or, in Wigner's setup, the eigenspace of $\Lambda$ with eigenvalue 1). Here is a counter-example: Let $\kilbert=\ell^2$, and let $A=C$ be complex conjugation (so the standard basis $e_i = (\delta_{ni})_{n\in\NNN} \in \ell^2$ actually has the desired property). When Wigner asks us in (13) for a vector $v\neq 0$, choose $-ie_2$, so he sets $v_{11}=e_2$ in (15). When Wigner asks us for a vector $0\neq v'\perp v_{11}$, choose $v'=e_4$, so he sets $v_{12}=e_4$; and so on. Then the sequence $(v_{11},v_{12},v_{13},\ldots)$ is $(e_2,e_4,e_6,\ldots)$, which is not an ONB of $\kilbert$. The same counter-example also applies to Ruotsalainen's proof.} so we give here a proof for the sake of completeness:

\begin{proof}[Proof of Proposition~\ref{prop:C}]
    It suffices to show that there exists an ONB $(f_j)_j$ with $Cf_j=f_j$ for every $j$; after all, by the conjugate linearity of $C$ we then obtain \eqref{Cfj} for every finite linear combination, and since the anti-unitarity implies that $C$ is bounded and therefore continuous, also for infinite linear combinations.

    We first show that for every $0\neq g\in\kilbert$, one of the following 3 possibilities occurs: (i)~$Cg$ is a multiple of $g$, or (ii)~$Cg\perp g$, or (iii)~for $\alpha$ being any of the two complex square roots of $\scp{g}{Cg}$, $u:=\alpha g+\alpha^*Cg$ and $v:=i\alpha g-i\alpha^* Cg$ are mutually orthogonal nonzero vectors with $Cu=u$ and $Cv=v$.

    Indeed, if neither (i) nor (ii) occurs, then $\alpha\neq 0$ by non-(ii); furthermore, $g$ and $Cg$ are linearly independent by non-(i); thus, $u$ and $v$ are both non-zero, one easily verifies that $Cu=u$ and $Cv=v$, and
    \begin{align}
        \scp{u}{v}&=\scp{\alpha g+\alpha^*Cg}{i\alpha g-i\alpha^* Cg}\\[3mm]
        &=i|\alpha|^2 \|g\|^2+i\alpha^2 \scp{Cg}{g}-i\alpha^{*2}\scp{g}{Cg}-i|\alpha|^2 \underbrace{\|Cg\|^2}_{=\|g\|^2}\\
        &=i|\scp{g}{Cg}|^2-i|\scp{g}{Cg}|^2=0\,.
    \end{align}
    This proves the trichotomy mentioned.
    
    Now let $(g_i)_i$ be any ONB of $\kilbert$. If $Cg_1$ is a multiple of $g_1$, say $Cg_1=z\, g_1$ with $z\in\CCC$, then $|z|=1$ by anti-unitarity, say $z=e^{i\theta}$. Set $f_1:= e^{i\theta/2}g_1$. Then $\|f_1\|=\|g_1\|=1$ and $Cf_1=e^{-i\theta/2}Cg_1=e^{-i\theta/2}e^{i\theta}g_1=f_1$. If, however, $Cg_1\perp g_1$, then set
    \begin{align}
        f_1&:=\frac{1}{\sqrt{2}}(g_1+Cg_1)\\
        f_2&:=\frac{i}{\sqrt{2}}(g_1-Cg_1)\,.
    \end{align}
    One easily verifies that $\|f_1\|=1=\|f_2\|$, $\scp{f_1}{f_2}=0$, $Cf_1=f_1$, and $Cf_2=f_2$. If, however, $Cg_1$ is neither parallel nor orthogonal to $g_1$, then introduce $u$ and $v$ as above and set $f_1=u/\|u\|$ and $f_2=v/\|v\|$. We have thus proved the anchor of the induction with the induction hypothesis for $n\in\NNN$ that \emph{for $V_n:=\mathrm{span}\{g_1,Cg_1,\ldots, g_n,Cg_n\}$, there is an ONB $(f_1,\ldots,f_{\dim V_n})$ with $Cf_j=f_j$ for all $j\in\{1,\ldots,\dim V_n\}$.}

    For the induction step, if $g_{n+1}\in V_n$, then also $Cg_{n+1}\in V_n$ because $V_n$ is closed under $C$, $CV_n\subseteq V_n$, so there is nothing to prove. If $g_{n+1}\notin V_n$, then decompose $g_{n+1}$ into the part in $V_n$ and the part orthogonal to $V_n$; call the latter $\tilde{g}_{n+1}\neq 0$, set $g:=\tilde{g}_{n+1}/\|\tilde{g}_{n+1}\|$, and apply the same reasoning to $g$ as above (for the induction anchor) to $g_1$; this yields an ONB of $\mathrm{span}\{g,Cg\}$ that we call $f_{\dim V_n +1}$ and $f_{\dim V_n+2}$ (or, if $Cg\propto g$, just $f_{\dim V_n+1}$). Since $CV_n\subseteq V_n$ and $g\perp V_n$, also $Cg\perp V_n$, so $\mathrm{span}\{g,Cg\}\perp V_n$ and the new $f$ vectors are $\perp V_n$ as well. One easily verifies that $\mathrm{span}\{g_1,Cg_1,\ldots,g_n,Cg_n,g,Cg\}=\mathrm{span}\{g_1,Cg_1,\ldots,g_n,Cg_n,g_{n+1},Cg_{n+1}\}$. This completes the proof of the induction.

    We thus have that every $V_n$ possesses an ONB $(f_1,\ldots, f_{\dim V_n})$ with $Cf_j=f_j$; it remains to prove the same for $\kilbert$ in the case $\dim\kilbert=\infty$. Here, it is relevant that the induction step added further $f$ vectors but did not change any of the previous ones. We thus obtain a sequence $(f_1,f_2,\ldots)$ of orthonormal vectors such that $\mathrm{span}\{f_1,\ldots,f_{\dim V_n}\}=V_n$. It follows that
    \be
    \mathrm{span}\{f_1,f_2,\ldots\}=\bigcup_{n=1}^\infty V_n=\mathrm{span}\{g_1,Cg_1,g_2,Cg_2,\ldots\} \supseteq \mathrm{span}\{g_1,g_2,\ldots\}\,,
    \ee
    which is a dense subspace of $\kilbert$. Therefore, $(f_j)_j$ is an ONB of $\kilbert$.
\end{proof}

\subsection{Proof of Theorem~\ref{TheoremVacuum}}
\label{sec:thm1}
 
Theorem \ref{TheoremVacuum} asserts that even if the series \eqref{QAdef} for $Q_A$ converged in some sense, states obtained through polynomials of the fields acting on the vacuum cannot belong to the domain of $Q_A$ for cube-shaped regions $A$. 
Before proving it, we prove a lemma in which we compute the norm of the charge operator acting on the vacuum, in terms of an orthonormal basis. Afterward, we will show that this norm is infinite. 

In the following, let $e_k$ for $k\in\{1,2,3,4\}$ denote the standard basis vectors of spin space $\CCC^4$. We also recall from \eqref{Qtruncated} the notation
\be
Q_{A,(f_j)}^J = \sum_{j=1}^J :\Psi^*(f_j) \Psi(f_j):
\ee
for the truncated charge operator.  
One would expect $\|Q_A\Omega\|^2$, if it is finite, to be equal to $\displaystyle \lim_{J\to\infty}\|Q_{A,(f_j)}^J \Omega\|^2$.

\begin{lemma}\label{LemmaVacuum}
Let $A$ be any Borel measurable subset of $\RRR^3$, $(\varphi_i)_{i \in \mathbb{N}}$ an ONB of $L^2(A,\CCC)$,  $(\sigma_i)_{i \in \mathbb{N}}$ an ONB of $L^2(A^c,\CCC)$, $f_i^k:=\varphi_i\otimes e_k$. Then
\begin{equation}\label{vacuumOrthonormalBasis}
 \liminf_{J\to\infty} \bigl\Vert Q_{A,(f_i^k)}^J \Omega \bigr\Vert^2 \geq  \sum_{i,j=1}^{\infty} \left[ m^2   \left\vert \left\langle \hat{\sigma}_i, \frac{1}{\lambda(\vp)}\hat{\varphi}_j \right\rangle \right\vert^2
    +  \sum_{s=1}^3 \left\vert \left\langle \hat{\sigma}_i, \frac{p_s}{\lambda(\vp)}\hat{\varphi}_j \right\rangle \right\vert^2 \right] ,
\end{equation}
where
\be\label{lambdadef}
\lambda(\vp)=\sqrt{\Vert \vp \Vert^2+m^2}
\ee
is the energy function and $\hat{f}$ denotes the Fourier transform of the function $f$.
\end{lemma}

\begin{proof}
For any ONB $(f_j)$ of $L^2(A,\CCC^4)$, the only non-vanishing sector of the vector $Q_{A,(f_j)}^J \Omega$ is the $(1,1)$ sector given by
\begin{subequations}
\begin{align}
\bigl(Q_{A,(f_j)}^J \Omega \bigr)^{(1,1)}&=\sum_{i=1}^{J}b^*(f_i)c^*(f_i)\Omega \\
&=\sum_{i=1}^{J}P_+f_i \otimes CP_- f_i,
\end{align}
\end{subequations}
with norm
\begin{equation}\label{eq: NormChargeOperatorVacumm}
    \bigl\Vert Q_{A,(f_j)}^J \Omega \bigr\Vert^2
=\sum_{i,j=1}^{J} \langle P_+ f_i, f_j \rangle \langle f_j, P_-f_i \rangle.
\end{equation}
The projections $P_+$ and $P_-$ in the Fourier space have the form of a multiplication operator
\begin{equation}
\Lambda^{\pm}=\mathcal{F}^{\oplus 4}P_{\pm}(\mathcal{F}^{\oplus 4})^{-1}=\frac{1}{2}\pm\frac{m\beta + \va\cdot \vp}{2\lambda(\vp)},
\end{equation}
where $ \alpha_j=\gamma^0 \gamma^j$, $\beta=\gamma^0$ and $\mathcal{F^{\oplus 4}}$ denotes the direct sum of four Fourier transform operators (see \cite[p.~9,10]{Thaller}). Consider now the Fourier transformed bases $(\hat{\varphi}_i)_i$ and $(\hat{f}_i^k)_{i\in\NNN,k\in\{1\ldots 4\}}$; then we can rewrite \eqref{eq: NormChargeOperatorVacumm} as
\begin{subequations}\label{eq: NormChargeOperatorVacummMomentumSpace}
\begin{align}
\bigl\Vert Q_{A,(f_i^k)}^J \Omega \bigr\Vert^2&=\sum_{i,j=1}^{J}\sum_{k,l=1}^4 \langle \hat{f}_i^k, \Lambda^+ \hat{f}_j^l \rangle \langle \hat{f}_j^l, \Lambda^- \hat{f}_i^k \rangle\\
&=\frac{1}{4}\sum_{i,j=1}^{J}\sum_{k,l=1}^4 \left[ \delta_{ij}\delta_{kl} -\left\vert  \left\langle \hat{f}_i^k, \frac{m\beta+\va\cdot \vp }{\lambda(\vp)}\hat{f}_j^l \right\rangle  \right\vert^2 \right].
\end{align}
\end{subequations}
By direct computation, one can obtain the following relations:
\begin{subequations}
\begin{align}
    \left\vert \left\langle \hat{f}_i^k, \frac{1}{\lambda(\vp)}\beta \hat{f}_j^l \right\rangle \right\vert &= \delta_{kl} \left\vert \left\langle \hat{\varphi}_i, \frac{1}{\lambda(\vp)}\hat{\varphi}_j \right\rangle \right\vert,\\
    \sum_{k,l=1}^4  \left\vert \left\langle \hat{f}_i^k, \frac{\va\cdot \vp}{\lambda(\vp)} \hat{f}_j^l \right\rangle \right\vert^2 &=4 \sum_{s=1}^3 \left\vert \left\langle \hat{\varphi}_i, \frac{p_s}{\lambda(\vp)}\hat{\varphi}_j \right\rangle \right\vert^2,\\
    \mathrm{Re}\left[ \overline{ \left\langle \hat{f}_i^k, \frac{\beta}{\lambda(\vp)}\hat{f}_j^l \right\rangle    }  \left\langle \hat{f}_i^k, \frac{\va\cdot \vp}{\lambda(\vp)}\hat{f}_j^l \right\rangle      \right]&=0,
    \end{align}
\end{subequations}
which, when applied in \eqref{eq: NormChargeOperatorVacummMomentumSpace}  and after summing over $k,l$, lead to the following expression:
\begin{equation}
\Vert Q_{A,(f_i^k)}^J \Omega \Vert^2=\sum_{i,j=1}^{J} \left[ \delta_{ij}-m^2\left\vert \left\langle \hat{\varphi}_i, \frac{1}{\lambda(\vp)}\hat{\varphi}_j \right\rangle \right\vert^2-  \sum_{s=1}^3 \left\vert \left\langle \hat{\varphi}_i, \frac{p_s}{\lambda(\vp)}\hat{\varphi}_j \right\rangle \right\vert^2    \right].
\end{equation}

Now we use the ONB $(\sigma_i)$ of $L^2(A^c)$ and the fact that $1/\lambda(\vp)$ and $p_s/\lambda(\vp)$ are bounded functions. Since the union $(\varphi_i,\sigma_i)_{i \in \mathbb{N}}$ is an ONB of $L^2(\RRR^3)$,  
Parseval's identity states that for every $j \in \mathbb{N}$,
\begin{equation}
\sum_{i=1}^{\infty}\left[ \left\vert \left\langle \hat{\varphi}_i,\frac{1}{\lambda(\vp)}\hat{\varphi}_j \right\rangle \right\vert^2+ \left\vert \left\langle \hat{\sigma}_i,\frac{1}{\lambda(\vp)}\hat{\varphi}_j \right\rangle \right\vert^2 \right]=\left\Vert \frac{1}{\lambda(\vp)} \hat{\varphi}_j \right\Vert^2
\end{equation}
and
\begin{equation}
\sum_{i=1}^{\infty}\left[ \left\vert \left\langle \hat{\varphi}_i,\frac{p_s}{\lambda(\vp)}\hat{\varphi}_j \right\rangle \right\vert^2+ \left\vert \left\langle \hat{\sigma}_i,\frac{p_s}{\lambda(\vp)}\hat{\varphi}_j \right\rangle \right\vert^2 \right]=\left\Vert \frac{p_s}{\lambda(\vp)} \hat{\varphi}_j \right\Vert^2
\end{equation}
for every $s \in \{1,2,3\}$. As a consequence,
\begin{subequations}
\begin{align}
    \bigl\Vert Q_{A,(f_i^k)}^J \Omega \bigr\Vert^2
    &=\sum_{j=1}^J\left[ \Vert \hat{\varphi}_j \Vert^2-m^2 \left\Vert \frac{1}{\lambda(\vp)}\hat{\varphi}_j \right\Vert^2- \sum_{s=1}^3\left\Vert \frac{p_s}{\lambda(\vp)} \hat{\varphi}_j \right\Vert^2 \right.\nonumber\\ 
    &\hspace{20mm} +\sum_{i=J+1}^\infty \left(m^2  \left\vert \left\langle  \hat{\varphi}_i,\frac{1}{\lambda(\vp)}\hat{\varphi}_j \right\rangle \right\vert^2+\sum_{s=1}^3\left\vert \left\langle  \hat{\varphi}_i,\frac{p_s}{\lambda(\vp)}\hat{\varphi}_j \right\rangle \right\vert^2 \right)  \nonumber\\ 
    &\hspace{20mm} + \left. \sum_{i=1}^{\infty}\left(m^2  \left\vert \left\langle  \hat{\sigma}_i,\frac{1}{\lambda(\vp)}\hat{\varphi}_j \right\rangle \right\vert^2+\sum_{s=1}^3\left\vert \left\langle  \hat{\sigma}_i,\frac{p_s}{\lambda(\vp)}\hat{\varphi}_j \right\rangle \right\vert^2 \right)\right]\label{eq:SumLemmaChargeoperator}\\
    &\geq \sum_{j=1}^J\Biggl[ \biggl\langle \hat{\varphi}_j, \underbrace{\biggl(1-\frac{m^2}{\lambda(\vp)^2}-  \frac{\|\vp\|^2}{\lambda(\vp)^2} \biggr)}_{=0 \text{ by \eqref{lambdadef}}}\hat{\varphi}_j \biggr\rangle \nonumber\\ 
    &\hspace{20mm} + \left. \sum_{i=1}^{\infty}\left(m^2  \left\vert \left\langle  \hat{\sigma}_i,\frac{1}{\lambda(\vp)}\hat{\varphi}_j \right\rangle \right\vert^2+\sum_{s=1}^3\left\vert \left\langle  \hat{\sigma}_i,\frac{p_s}{\lambda(\vp)}\hat{\varphi}_j \right\rangle \right\vert^2 \right)\right],
\end{align}
\end{subequations}
which implies \eqref{vacuumOrthonormalBasis}.
\end{proof}

In order to compute the first term in  expression \eqref{vacuumOrthonormalBasis},  we will need the next result to express the function $\lambda(\vp)^{-1}$ in the position space.

\begin{lemma}\label{FunctionLambda}
    Let again $\lambda(\vp)=\sqrt{\Vert \vp \Vert^2 +m^2}$. Then
    \begin{equation}
        \mathcal{F}^{-1}\left( \frac{1}{\lambda} \right)(\vx)=m\sqrt{\frac{2}{\pi}}\frac{K_1(m \Vert \vx \Vert)}{\Vert \vx \Vert}, \quad \vx \in \mathbb{R}^3,
    \end{equation}
distributionally, where $K_1$ is a modified Bessel function of the second kind.
\end{lemma}

\begin{proof}
Following the ideas of \cite[p.~65]{Electro}, we compute the distributional limit
\begin{subequations}
\begin{align}
(2\pi)^{3/2}\mathcal{F}^{-1}\left( \frac{1}{\lambda} \right)(\vx)
&=\lim_{R \to \infty}\int_0^R \mathrm{d} \Vert \vp \Vert\int_0^{\pi}  \mathrm{d} \theta \int_0^{2\pi}  \mathrm{d}\varphi \frac{e^{i \Vert \vx \Vert \Vert \vp \Vert \cos \theta}}{\sqrt{\Vert \vp \Vert^2+m^2}} \Vert \vp \Vert^2 \sin \theta \\
&=2\pi\lim_{R \to \infty}\int_0^R \mathrm{d} \Vert \vp \Vert \frac{\Vert \vp \Vert^2  \, }{\sqrt{\Vert \vp \Vert^2+m^2}} 
 \int_0^{\pi}e^{i \Vert \vx \Vert \Vert \vp \Vert \cos \theta}  \sin \theta \, \mathrm{d} \theta\\
&= 4\pi \lim_{R \to \infty} \frac{1}{\Vert \vx \Vert} \int_0^R \frac{\sin(\Vert \vp \Vert \Vert \vx \Vert) \Vert \vp \Vert}{\sqrt{\Vert \vp \Vert^2+m^2}} \mathrm{d}\Vert \vp \Vert\\
&=-4\pi  \lim_{R \to \infty} \frac{1}{\Vert \vx \Vert} \frac{\mathrm{d}}{\mathrm{d} \Vert \vx \Vert} \int_0^R \frac{\cos(m\Vert \vp \Vert \Vert \vx \Vert)}{\sqrt{\Vert \vp \Vert^2+1}} \mathrm{d}\Vert \vp \Vert.
\end{align}
\end{subequations}
Now, let
\begin{equation}
f_R(m\Vert \vx \Vert)=\int_0^R \frac{\cos(m\Vert \vp \Vert \Vert \vx \Vert)}{\sqrt{\Vert \vp \Vert^2+1}} \mathrm{d}\Vert \vp \Vert,
\end{equation}
and 
\begin{equation}
K_0(m\Vert \vx \Vert)=\int_0^{\infty} \frac{\cos(m\Vert \vp \Vert \Vert \vx \Vert)}{\sqrt{\Vert \vp \Vert^2+1}} \mathrm{d}\Vert \vp \Vert.
\end{equation}
Due to the fact that $f_R \to K_0$ on compact sets,  they converge distributionally, and using that $K_0'=-K_1$ (see, e.g., \cite[p.~376]{Bessel1}), we conclude that
\begin{equation}
(2\pi)^{3/2}\mathcal{F}^{-1}\left( \frac{1}{\lambda} \right)(\vx)=-\frac{4\pi}{\Vert \vx  \Vert}\frac{\mathrm{d}}{\mathrm{d} \Vert \vx \Vert}K_0(m \Vert \vx \Vert)=4 \pi m \frac{K_1(m \Vert \vx \Vert)}{\Vert \vx \Vert}.
\end{equation}
\end{proof}

The last preliminary result that we will need is  to prove Theorem \ref{TheoremVacuum} for the particular case where $\psi=|\Omega\rangle$. This is the purpose of the next lemma.

\begin{lemma}\label{lem3}
    Let $\vx\in\RRR^3$, and let $A=\vx+[-\pi,\pi]^3\subset \RRR^3$. Then there is an ONB $(f_j)_j$ of $L^2(A,\CCC^4)$ such that
\be
\lim_{J\to\infty} \Bigl\| Q_{A,(f_j)}^J ~ |\Omega\rangle \Bigr\|^2 = \infty\,.
\ee
\end{lemma}

\begin{proof}
    We will show that the first summand in \eqref{vacuumOrthonormalBasis} diverges. Since the Lebesgue measure is invariant under translation, we can assume without loss of generality that $\vx=0$. Consider the ONB given by
    \be
    \varphi_{\vk}(\vx)=\frac{1}{(2\pi)^{3/2}}e^{i\vk \cdot \vx}\mathcal{\chi}_A(\vx), \quad \vk\in \mathbb{Z}^3,
    \ee
where $\chi_A$ is the characteristic function on $A=[-\pi,\pi]^3$. Let $(\sigma_j)_j$ be an ONB of $L^2(A^c)=L^2(A^c,\CCC)$. Using the notation of  Lemma \ref{LemmaVacuum},  Lemma \ref{FunctionLambda}, and the Convolution Theorem for tempered distributions (see, e.g., \cite[Thm. IX.4]{Simon}), we get that
\begin{subequations}
\begin{align}
    \sum_{i,j=1}^{\infty} \left\vert \left\langle  \hat{\sigma}_i, \frac{1}{\lambda(\vp)}\hat{\varphi}_j \right\rangle \right\vert^2 &= \nonumber\\ & \hspace{-1cm} = \sum_{\vk \in \mathbb{Z}^3}\sum_{j=1}^{\infty} \left\vert \left\langle \sigma_j, m \sqrt{\frac{2}{\pi}}\frac{K_1(m\Vert \vx \Vert)}{\Vert \vx \Vert}*\varphi_{\vk}(\vx) \right\rangle \right\vert^2 
\\
    & \hspace{-1cm}=\sum_{\vk \in \mathbb{Z}^3}\int_{A^c} \left\vert  m \sqrt{\frac{2}{\pi}}\frac{K_1(m\Vert \vx \Vert)}{\Vert \vx \Vert}*\varphi_{\vk}(\vx) \right\vert^2 \mathrm{d}\vx 
\\
    & \hspace{-1cm} =\sum_{\vk \in \mathbb{Z}}\int_{A^c} \left\vert \int_{\mathbb{R}^3}   m \sqrt{\frac{2}{\pi}}\frac{K_1(m\Vert \vy \Vert)}{\Vert \vy \Vert}   \chi_{[-\pi,\pi]^3}(\vx-\vy) \frac{e^{i \vk \cdot \vy}}{(2\pi)^{3/2}} \mathrm{d}\vy  
    \right\vert^2 \mathrm{d}\vx    
\\ &\hspace{-1cm}=
  \int_{A^c} \left\Vert m \sqrt{\frac{2}{\pi}}\frac{K_1(m\Vert \vy \Vert)}{\Vert \vy \Vert} \right\Vert^2_{L^2(\vx+[-\pi,\pi]^3)}  \mathrm{d}\vx
\\ & \hspace{-1cm} \geq 
    \int_{S_1} \left\Vert m \sqrt{\frac{2}{\pi}}\frac{K_1(m\Vert \vy \Vert)}{\Vert \vy \Vert} \right\Vert^2_{L^2(\vx+[-\pi,\pi]^3)}\mathrm{d}\vx \, ,
\end{align}
\end{subequations}
where 
\begin{equation}
    S_1=\bigl\{\vx \in A^c : \Vert \vx \Vert \leq \sqrt{3}\pi+1 \bigr\},
\end{equation}
see Figure~\ref{fig:S1}.
Now, in \cite[p.~375]{Bessel1} it can be found that
\begin{equation}
K_1(m \Vert \vy \Vert) \sim_0 \frac{1}{2}\Gamma(1)\frac{2}{m \Vert \vy \Vert}= \frac{1}{m \Vert \vy \Vert}, 
\end{equation}
where $f(y) \sim_0 g(y)$ means $\displaystyle \lim_{y \to 0}f(y)g(y)^{-1}=1$.
Thus, for every $0 < \varepsilon <1$, there exists $\delta>0$ such that if $\Vert \vy \Vert < \delta$, then
\begin{equation}
    \left\vert \frac{K_1(m\Vert \vy \Vert)}{m \Vert \vy \Vert}-\frac{1}{(m \Vert \vy \Vert)^2} \right\vert< \frac{\varepsilon}{(m \Vert \vy \Vert)^2},
\end{equation}
i.e.,
\begin{equation}
    \frac{1-\varepsilon}{(m \Vert \vy \Vert)^2}< \frac{K_1(m\Vert \vy \Vert)}{m \Vert \vy \Vert}<\frac{1+\varepsilon}{(m \Vert \vy \Vert)^2}.
\end{equation}
Without loss of generality, we can assume $\delta<1$. Define  for every $\vx \in S_1$ the set $A_{\vx}=(\vx+[-\pi,\pi]^3)\cap B_0(\delta)$ (see Figure~\ref{fig:SetAx}) and also
\begin{equation}
    S_2=S_1 \cap \bigl\{ x\in A^c: \vol(A_{\vx}) \neq 0 \bigr\} .
\end{equation}
Then
\begin{equation}
    \int_{S_1} \left\Vert m \sqrt{\frac{2}{\pi}}\frac{K_1(m\Vert \vy \Vert)}{\Vert \vy \Vert} \right\Vert^2_{L^2(\vx+[-\pi,\pi]^3)} \mathrm{d}\vx \geq \frac{2(1-\varepsilon)^2}{\pi 
 }\int_{S_2}\mathrm{d}\vx\int_{A_{\vx}} \mathrm{d}\vy\frac{1}{\Vert \vy \Vert^4} \, .\label{100}
\end{equation}
Now, since we assume  $\vol(A_{\vx}) \neq 0$ and $\delta <1$, there are three possibilities for the intersection between the cube and the sphere:
\begin{itemize}
\item[1.] The cube slice inside the sphere intersects a coordinate axis.
\item[2.] The cube slice inside the sphere intersects a coordinate plane but not a coordinate axis.
\item[3.] If 1 and 2 are not satisfied then there is just one single vertex of the cube in the
intersection.
\end{itemize}

\begin{center}
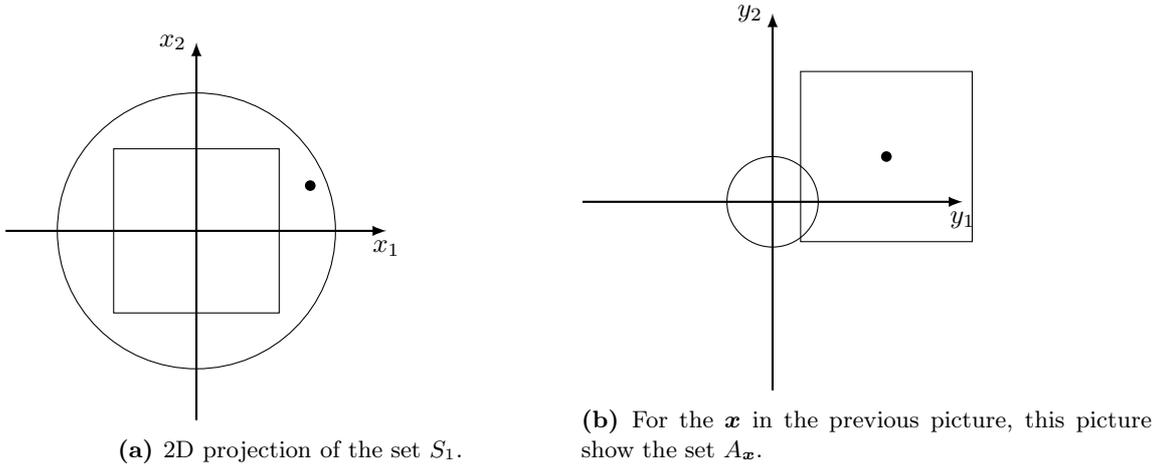
\begin{figure}[H]
\begin{subfigure}{.5\linewidth}
\begin{tikzpicture}[line cap=round]
    \draw[thick, -latex] (-2.5,0) -- (2.5,0) node[below] {$x_1$};
    \draw[thick, -latex] (0,-2.5) -- (0,2.5) node[left] {$x_2$};
		
    \newcommand\Square[1]{+(-#1,-#1) rectangle +(#1,#1)}
    \draw [very thin, lightgray] (0,0) ;  
    \draw (0,0) +(-5.4413980927/5,-5.4413980927/5) rectangle +(5.4413980927/5,5.4413980927/5) ;
    \draw[fill=none] (0,0) 
    circle (1.829cm) 
    node [black]{}; 	
    \fill (3.4413980927/3+0.35,0.6) circle[radius=2pt]{};
  
\end{tikzpicture}
\caption{2D projection of the set $S_1$.}
\label{fig:S1}
\end{subfigure}
\begin{subfigure}{.5\linewidth}
 \begin{tikzpicture}
    \draw[thick, -latex] (-2.5,0) -- (2.5,0) node[below] {$y_1$};
    \draw[thick, -latex] (0,-2.5) -- (0,2.5) node[left] {$y_2$};
\fill (3.4413980927/3+0.35,0.6) circle[radius=2pt]{};
  \draw[name path = A] (3.4413980927/3+0.35,0.6) +(-5.6413980927/5,-5.6413980927/5) rectangle +(5.6413980927/5,5.6413980927/5);
   \draw[name path = B](0,0) circle (0.6) node [black,yshift=-1.5cm]{} ;


\end{tikzpicture}
\caption{For the $\vx$ in the previous picture, this picture show the set $A_{\vx}$.}
\label{fig:SetAx}
\end{subfigure}
\caption{Graphical representation of the sets considered in the proof.}
\end{figure}
\end{center}

We will show that for the first case, the integral diverges and we will assume that the cube cuts the $x_3$ axis. For $\vx$ in
\begin{equation}
G=\left\{\vx \in S_1: x_3> \pi, x_1^2+x_2^2\leq \frac{\delta^2}{100} \right\},
\end{equation}
$A_{\vx}$ has the shape of a filled lens as shown in Figure \ref{fig:SetAx}. Expressing $\vy\in A_{\vx}$ in spherical coordinates, this lens is characterized by the conjunction of the two conditions $r\leq \delta$ and
\begin{equation}
r\cos \theta\geq x_3-\pi \,.
\end{equation}
For dealing with the latter, we note that for points in the intersection between the plane $r \cos\theta = x_3-\pi$
and the sphere $r=\delta$, we have that
$\theta = \arccos ((x_3-\pi)/\delta)$. Therefore, continuing from the right-hand side of \eqref{100},
\begin{subequations}
\begin{align}
    \frac{2(1-\varepsilon)^2}{\pi 
 }\int_{S_2}\mathrm{d}\vx\int_{A_{\vx}} \mathrm{d}\vy\frac{1}{\Vert \vy \Vert^4}  \geq\nonumber\\[3mm]
 & \hspace{-5cm} \geq \frac{2(1-\varepsilon)^2}{\pi 
 }\int_{G }\mathrm{d}\vx\int_{A_{\vx}} \mathrm{d}\vy \frac{1}{\Vert \vy \Vert^4} \\[3mm]
    &\hspace{-5cm}\geq \frac{2(1-\varepsilon)^2}{\pi}\int_G \mathrm{d}\vx \, 2\pi \hspace{-3mm} \int\limits_0^{\arccos((x_3-\pi)/\delta)}\hspace{-3mm}\mathrm{d}\theta\int_{
    \frac{x_3-\pi}{\cos \theta} }^{\delta} \mathrm{d}r \frac{\sin \theta}{r^2}  \\[3mm]
    &\hspace{-5cm}= 4(1-\varepsilon)^2\int_G \mathrm{d}\vx \hspace{-3mm} \int\limits_0^{\arccos((x_3-\pi)/\delta) }\hspace{-3mm}\mathrm{d}\theta \, \sin \theta \, \biggl(-\frac{1}{\delta}+\frac{\cos \theta}{x_3-\pi} \biggr)  \\[3mm]   &\hspace{-5cm} = 4(1-\varepsilon)^2\left[   -\int_G \frac{1}{\delta}\Bigl(1-\frac{x_3-\pi}{\delta}\Bigr) \, \mathrm{d}\vx + \int_G \frac{1}{2(x_3-\pi)}\mathrm{d}\vx -\frac{1}{2\delta^2}\int_G(x_3-\pi) \, \mathrm{d}\vx \right].
\end{align}
\end{subequations}

The first and third integral are clearly bounded, so we have to compute the second. $G$ is a
cylinder bounded on the bottom by the cube and on the top by the sphere, and since the function we are integrating is positive, we can lower bound the integral by integrating up
to the height at which the cylinder and the sphere intersect, i.e.,  $\pi \leq x_3 \leq \alpha$, where 
\begin{equation}
\alpha=\sqrt{(\sqrt{3}\pi+1)^2-\frac{\delta^2}{100}}.
\end{equation}
Then,
\begin{equation}
    \int_G \frac{1}{2(x_3-\pi)}\mathrm{d}\vx \geq \int_{\pi}^{\alpha} \frac{1}{2(x_3-\pi)} \pi \frac{\delta^2}{100}\mathrm{d}x_3,
\end{equation}
which is divergent, and the proof of Lemma~\ref{lem3} is finished.
\end{proof}

\underline{\textit{Proof of Theorem \ref{TheoremVacuum}:}}

Let $\psi \in F$, set  $L := \max \{n+m : \psi^{(n,m)}\neq 0\}$ 
and choose $(n_0,m_0)$ such that $n_0+m_0=L$ and $\psi^{(n_0,m_0)} \neq 0$. Then
\be
(Q_{A,(f_i)}^J \psi)^{(n_0+1,m_0+1)} = \sum_{i=1}^J b^*(f_i) c^*(f_i) \psi^{(n_0,m_0)}.
\ee
Since $\psi \in F$, we can write $\psi^{(n_0,m_0)}$ as a finite linear combination of vectors with $n_0$ particles  and $m_0$ anti-particles, i.e.,
\be\label{psiphi}
\psi^{(n_0,m_0)}=\sum_{s=1}^S \varphi_s
\ee
with
\be
\varphi_s= b^*(g_1^s)\hdots b^*(g_{n_0}^s)c^*(h_1^s)\hdots c^*(h_{m_0}^s)\Omega\neq 0.
\ee
As a consequence, we can lower bound 
\begin{subequations}
\begin{align}
    \bigl\Vert Q_{A,(f_j)}^J \psi \bigr\Vert^2 
    & \geq  \bigl\Vert \bigl(Q_{A,(f_j)}^J \psi \bigr)^{(n_0+1,m_0+1)}\bigl\Vert^2\\
     &\geq \sum_{i,j=1}^{J}\Bigl\langle b^*(f_i)c^*(f_i)\psi^{(n_0,m_0)},b^*(f_j)c^*(f_j)\psi^{(n_0,m_0)} \Bigr\rangle\\ 
    & =\sum_{i,j=1}^{J}\langle f_i,P_+ f_j \rangle  \langle c^*(f_i)\psi^{(n_0,m_0)},c^*(f_j)\psi^{(n_0,m_0)}\rangle \: + \nonumber\\
    & +\sum_{i,j=1}^{J}\langle c^*(f_i)\psi^{(n_0,m_0)},b^*(f_j)c^*(f_j)b(f_i)\psi^{(n_0,m_0)} \rangle \\
    & =\sum_{i,j=1}^{J}  \langle f_i,P_+ f_j \rangle \langle f_j, P_- f_i \rangle \Vert \psi^{(n_0,m_0)}\Vert^2 - \nonumber\\
    & - \sum_{i,j=1}^{J}\langle f_i,P_+ f_j \rangle  \langle \psi^{(n_0,m_0)},c^*(f_j)c(f_i)\psi^{(n_0,m_0)} \rangle \: - \nonumber\\
    &  - \sum_{i,j=1}^{J}\langle f_i, P_- f_j\rangle \langle \psi^{(n_0,m_0)}, b^*(f_j) b(f_i) \psi^{(n_0,m_0)}\rangle \: + \nonumber\\
    & +\sum_{i,j=1}^{J}\langle b(f_j)c(f_j)\psi^{(n_0,m_0)}, b(f_i)c(f_i)\psi^{(n_0,m_0)} \rangle \label{eq:FieldsCase1}
\end{align}
\end{subequations}
using the anti-commutation relations \eqref{bcanticommute}. 
Notice that the first sum in \eqref{eq:FieldsCase1} is equal to
\be
\Vert Q_{A,(f_j)}^J \Omega \Vert^2 \Vert \psi^{(n_0,m_0)} \Vert^2, 
\ee
so if we prove that the other three sums converge as $J\to\infty$, then Theorem~\ref{TheoremVacuum} follows. To this end, introduce  the vectors
\begin{equation}
    \varphi_s(k)=b^*(g_1^s)\hdots b^*(g_{n_0}^s)c^*(h_1^s)\hdots c^*(h_{k-1}^s)c^*(h_{k+1}^s) \hdots  c^*(h_{m_0}^s)\Omega
\end{equation}
and 
\begin{equation}
    \varphi_s(u,k)=b^*(g_1^s)\hdots b^*(g_{u-1}^s)b^*(g_{u+1}^s)\hdots b^*(g_{n_0}^s)c^*(h_1^s)\hdots c^*(h_{k-1}^s)c^*(h_{k+1}^s) \hdots  c^*(h_{m_0}^s)\Omega.
\end{equation}
with $1 \leq k \leq  m_0$, $1\leq u \leq n_0$.   Let us turn to the second sum in \eqref{eq:FieldsCase1}. To take the limit $J\to\infty$ means to consider the series
\begin{subequations}
\begin{align}
   \sum_{i,j=1}^{\infty}\sum_{r,s=1}^S\langle f_i,P_+ f_j \rangle  \langle \varphi_r,c^*(f_j)c(f_i)\varphi_s \rangle  & =\nonumber \\
   & \hspace{-5cm} =  \sum_{i,j=1}^{\infty}\sum_{r,s=1}^S\langle f_i,P_+ f_j \rangle  \left\langle \sum_{k=1}^{m_0} (-1)^{k+1} \langle h_k^r,P_- f_j \rangle \varphi_r(k),\sum_{l=1}^{m_0} (-1)^{l+1} \langle h_l^s,P_- f_i \rangle \varphi_s(l)\right\rangle \\
    & \hspace{-5cm} =  \sum_{k,l=1}^{m_0} \sum_{r,s=1}^S  (-1)^{k+l}  \langle h_l^s, P_-P_AP_+P_AP_- h_k^r \rangle  \langle \varphi_r(k), \varphi_s(l) \rangle\,,   
\end{align}
\end{subequations}
which is convergent.
The third sum in  \eqref{eq:FieldsCase1} can be treated in the same way. Finally, the last sum in \eqref{eq:FieldsCase1} can be treated by an analogous reasoning (using again the anti-commutation relations \eqref{bcanticommute}) as follows: Taking $J\to\infty$ amounts to considering the series
\begin{subequations}
\begin{align}
     \sum_{i,j=1}^{\infty} \sum_{r,s=1}^S \langle b(f_j)c(f_j)\varphi_r, b(f_i)c(f_i)\varphi_s \rangle  &= \nonumber\\
       & \hspace{-6cm} =  \sum_{i,j=1}^{\infty}  \sum_{r,s=1}^S\left\langle \sum_{k=1}^{m_0} (-1)^{k+1} \langle h_k^r,P_- f_j \rangle b(f_j)\varphi_r(k),\sum_{l=1}^{m_0} (-1)^{l+1} \langle h_l^s,P_- f_i \rangle b(f_i)\varphi_s(l)\right\rangle \\
        & \hspace{-6cm} = \sum_{i,j=1}^{\infty} \sum_{r,s=1}^S \sum_{k,l=1}^{m_0} (-1)^{k+l}  \langle f_j,P_- h_k^r \rangle  \langle h_l^s,P_- f_i \rangle  \left\langle  b(f_j)\varphi_r(k), b(f_i)\varphi_s(l)\right\rangle \\
         & \hspace{-6cm} =  \sum_{i,j=1}^{\infty} \sum_{r,s=1}^S \sum_{k,l=1}^{m_0} \sum_{u,v=1}^{n_0} (-1)^{u+v} (-1)^{k+l}  \langle f_j,P_- h_k^r \rangle  \langle h_l^s,P_- f_i \rangle \langle g_u^r,P_+ f_j \rangle \langle f_i,P_+ g_v^s \rangle  \langle \varphi_r(u,k), \varphi_s(v,l) \rangle \\
         &\hspace{-6cm} =  \sum_{r,s=1}^S  \sum_{k,l=1}^{m_0} \sum_{u,v=1}^{n_0}  (-1)^{u+v} (-1)^{k+l} \langle g_u^r,P_+ P_A P_- h_k^r \rangle \langle  h_l^s,P_-P_A P_+ g_v^s \rangle\langle \varphi_r(u,k), \varphi_s(v,l) \rangle \,,
\end{align}
\end{subequations}
which is convergent.

\section{Conclusions}
\label{sec:conclusions}

We have studied the possibility of mathematically defining operators $Q_A$ acting on the standard Hilbert space of the free Dirac quantum field in Minkowski space-time and representing the observable given by the amount of charge within an arbitrary region $A$ in 3d physical space. We have described the difficulties we have encountered when trying to give a rigorous definition based on the intuitive formula \eqref{FormalChargeGeneral}. These difficulties suggest that no such definition of $Q_A$ exists. Of course, this leaves open whether a different approach towards defining $Q_A$ might succeed.

\vspace{1cm}

\textbf{Acknowledgments:} PCR acknowledges  funding from the Deutsche Forschungsgemeinschaft (DFG,
German Research Foundation) under Germany's Excellence Strategy-EXC2111-390814868.

\end{document}